\documentclass[a4paper, amsfonts, amssymb, amsmath, reprint, showkeys, twoside,superscriptaddress, onecolumn, nofootinbib]{revtex4-2}
\usepackage{amsmath,amsthm}
\usepackage{amssymb}
\usepackage{array}
\usepackage[export]{adjustbox}
\usepackage{rotating}
\usepackage{mathtools}
\usepackage{mathdots}
\usepackage{bm}
\usepackage{dsfont}
\usepackage{xcolor}
\usepackage[caption=false]{subfig}
\usepackage{float}
\makeatletter
\let\newfloat\newfloat@ltx
\makeatother
\usepackage{graphicx}

\usepackage{thmtools}
\usepackage{thm-restate}

\usepackage{mathrsfs}
\usepackage{nccmath}
\usepackage{physics}
\usepackage{comment}
\usepackage{ulem}
\usepackage[colorlinks=true, urlcolor=blue, linkcolor=blue, citecolor=blue]{hyperref}
\usepackage[capitalize]{cleveref}
\usepackage[acronym]{glossaries}
\setkeys{glslink}{hyper=false}
\usepackage{makecell}

\usepackage{tikz}
\usetikzlibrary{quantikz2}

\usepackage[compat=0.4]{yquant}
\useyquantlanguage{groups}

\newtheorem{theorem}{Theorem}

\usepackage{todonotes}

\newcommand{\bea}{\begin{eqnarray}}
\newcommand{\eea}{\end{eqnarray}}

\newacronym{qft}{QFT}{quantum Fourier transform}
\newacronym{qsp}{QSP}{quantum signal processing}
\newacronym{com}{COM}{centre of mass}
\newacronym[
  plural={DOFs},
  longplural={degrees of freedom}
]{dof}{DOF}{degree of freedom}
\newacronym{lcu}{LCU}{linear combination of unitaries}

\newcommand{\ceil}[1]{\left\lceil #1 \right\rceil}
\newcommand{\floor}[1]{\left\lfloor #1 \right\rfloor}

\usepackage{tikzit}

\tikzstyle{gate}=[shape=rectangle, text height=1.5ex, text depth=0.25ex, yshift=0mm, fill=white, draw=black, minimum height=5mm, minimum width=5mm, font={\small}, tikzit category=circuit]
\tikzstyle{big gate}=[shape=rectangle, text height=2.0ex, text depth=0.25ex, yshift=0mm, fill=white, draw=black, minimum height=20mm, minimum width=5mm, font={\small}, tikzit category=circuit]
\tikzstyle{Z dot}=[inner sep=0mm, minimum size=2mm, shape=circle, draw=black, fill={rgb,255: red,221; green,255; blue,221}, tikzit category=zx]
\tikzstyle{Z phase dot}=[minimum size=5mm, font={\footnotesize\boldmath}, shape=rectangle, rounded corners=2mm, inner sep=0.2mm, outer sep=-2mm, scale=0.8, tikzit shape=circle, draw=black, fill={rgb,255: red,221; green,255; blue,221}, tikzit draw=blue, tikzit category=zx]
\tikzstyle{X dot}=[Z dot, shape=circle, draw=black, fill={rgb,255: red,255; green,136; blue,136}, tikzit category=zx]
\tikzstyle{X phase dot}=[Z phase dot, tikzit shape=circle, draw=black, tikzit draw=blue, fill={rgb,255: red,255; green,136; blue,136}, font={\footnotesize\boldmath}, tikzit category=zx]
\tikzstyle{hadamard}=[fill=yellow, draw=black, shape=rectangle, inner sep=0.6mm, minimum height=1.5mm, minimum width=1.5mm, tikzit category=zx]
\tikzstyle{paulibox}=[fill={rgb,255: red,221; green,221; blue,255}, draw=black, shape=rectangle, inner sep=0.6mm, minimum height=5mm, minimum width=5mm, font={\footnotesize}, text height=1.5ex, text depth=0.25ex, tikzit category=zx]
\tikzstyle{vertex}=[inner sep=0mm, minimum size=1mm, shape=circle, draw=black, fill=black, tikzit category=misc]
\tikzstyle{vertex set}=[inner sep=0mm, minimum size=1mm, shape=circle, draw=black, fill=white, font={\footnotesize\boldmath}, tikzit category=misc]
\tikzstyle{small black dot}=[fill=black, draw=black, shape=circle, inner sep=0pt, minimum width=1.2mm, tikzit category=circuit]
\tikzstyle{cnot ctrl}=[fill=black, draw=black, shape=circle, inner sep=0pt, minimum width=1.2mm, tikzit category=circuit]
\tikzstyle{cnot targ}=[fill=white, draw=white, shape=circle, tikzit category=circuit, label={center:$\oplus$}, inner sep=0pt, minimum width=2.1mm, tikzit fill={rgb,255: red,102; green,204; blue,255}, tikzit draw=black]
\tikzstyle{ket}=[fill=white, draw=black, shape=regular polygon, regular polygon sides=3, regular polygon rotate=-30, scale=0.7, inner sep=1pt, tikzit category=circuit, tikzit shape=rectangle, tikzit fill=green]
\tikzstyle{bra}=[fill=white, draw=black, shape=regular polygon, regular polygon sides=3, regular polygon rotate=30, scale=0.7, inner sep=1pt, tikzit category=circuit, tikzit shape=rectangle, tikzit fill=red]
\tikzstyle{scalar}=[shape=rectangle, text height=1.5ex, text depth=0.25ex, yshift=-0.5mm, fill=white, draw=black, minimum height=5mm, minimum width=5mm, font={\small}]
\tikzstyle{clabel}=[fill=white, draw=none, shape=rectangle, tikzit fill={rgb,255: red,56; green,255; blue,242}, font={\footnotesize}, inner sep=1pt, tikzit category=labels]
\tikzstyle{empty diagram}=[draw={gray!40!white}, dashed, shape=rectangle, minimum width=1cm, minimum height=1cm, tikzit category=misc]
\tikzstyle{amap}=[fill=white, draw=black, shape=NEbox, tikzit category=asymmetric, tikzit fill=yellow, tikzit shape=rectangle]
\tikzstyle{amap conj}=[fill=white, draw=black, shape=NWbox, tikzit category=asymmetric, tikzit fill=green, tikzit shape=rectangle]
\tikzstyle{amap adj}=[fill=white, draw=black, shape=SEbox, tikzit category=asymmetric, tikzit fill=red, tikzit shape=rectangle]
\tikzstyle{amap trans}=[fill=white, draw=black, shape=SWbox, tikzit category=asymmetric, tikzit fill=orange, tikzit shape=rectangle]
\tikzstyle{astate}=[fill=white, draw=black, shape=NEtriangle, tikzit category=asymmetric, tikzit shape=circle, tikzit fill=yellow]
\tikzstyle{astate conj}=[fill=white, draw=black, shape=NWtriangle, tikzit category=asymmetric, tikzit shape=circle, tikzit fill=green]
\tikzstyle{astate adj}=[fill=white, draw=black, shape=SEtriangle, tikzit category=asymmetric, tikzit shape=circle, tikzit fill=red]
\tikzstyle{astate trans}=[fill=white, draw=black, shape=SWtriangle, tikzit category=asymmetric, tikzit shape=circle, tikzit fill=orange]
\tikzstyle{tri}=[fill={rgb,255: red,255; green,136; blue,136}, tikzit fill=red, draw=black, shape=regular polygon, regular polygon sides=3, regular polygon rotate=30, scale=0.7, inner sep=0.2pt, minimum size=7mm, tikzit shape=rectangle]
\tikzstyle{tri left}=[fill={rgb,255: red,255; green,136; blue,136}, tikzit fill=red, draw=black, shape=regular polygon, regular polygon sides=3, regular polygon rotate=210, scale=0.7, inner sep=0.2pt, minimum size=7mm, tikzit shape=rectangle]
\tikzstyle{med gate}=[fill=white, draw=black, shape=rectangle, minimum width=5mm, tikzit category=circuit, minimum height=15mm, font={\small}]
\tikzstyle{measurement}=[fill=white, draw=black, shape=rectangle, minimum width=0.5cm, minimum height=0.5cm, tikzit category=circuit, path picture={\draw[black] ([xshift=-2.5mm,yshift=-1.5mm]path picture bounding box.center) arc[start angle=180,end angle=0,radius=2.5mm]; \draw[black,->] ([yshift=-1.5mm]path picture bounding box.center) -- ([xshift=2mm,yshift=1.5mm]path picture bounding box.center);}]
\tikzstyle{control box}=[fill=white, draw=black, shape=rectangle, minimum height=2mm, minimum width=2mm, font={\small}, tikzit category=circuit, yshift=0mm]

\tikzstyle{hadamard edge}=[-, dashed, dash pattern=on 2pt off 0.5pt, thick, draw={rgb,255: red,68; green,136; blue,255}]
\tikzstyle{box edge}=[-, dash pattern=on 2pt off 0.5pt, thick, draw={rgb,255: red,203; green,192; blue,225}, dashed]
\tikzstyle{brace edge}=[-, tikzit draw=blue, decorate, decoration={brace,amplitude=1mm,raise=-1mm}]
\tikzstyle{diredge}=[->]
\tikzstyle{double edge}=[-, double, shorten <=-1mm, shorten >=-1mm, double distance=2pt]
\tikzstyle{gray edge}=[-, {gray!60!white}]
\tikzstyle{pointer edge}=[->, very thick, gray]
\tikzstyle{boldedge}=[-, line width=1.6pt, shorten <=-0.17mm, shorten >=-0.17mm]
\tikzstyle{bidir edge}=[<->, very thick, draw={rgb,255: red,191; green,191; blue,191}]

\begin{document}
\author{Danial Motlagh}
\affiliation{Xanadu, Toronto, ON, M5G 2C8, Canada}
\author{Matthew Pocrnic}
\affiliation{Xanadu, Toronto, ON, M5G 2C8, Canada}
\begin{abstract}
    Table lookup, often referred to as quantum read only memory (QROM), is one of the most widely used subroutines in quantum algorithms, and constitutes the majority share of algorithmic overheads in most practical applications of quantum computers. It involves the coherent loading of $N$ bitstrings of length $b$ in superposition, and naively has a non-Clifford cost of $N$ Toffolis. It is known that given access to $b\, \lambda$ dirty qubits, one can reduce the Toffoli cost of QROM to $2\frac{N}{\lambda} + 4b(\lambda - 1)$. In this work, we first present an optimization to reduce this cost to $2\frac{N}{\lambda} + 2b(\lambda - 1) + 2\lambda-6$ by replacing the ``SelectSwap" architecture with ``SelectCopy". We then provide a further optimization for the qubit-constrained regime where the Toffoli cost is typically $\sim 2\frac{N}{\lambda}$, and reduce it to $\sim (1+\frac{1}{b})\frac{N}{\lambda}$, cutting the cost by approximately $50\%$ and effectively matching the performance of clean-qubit QROM using dirty qubits for practical values of $b$. Lastly, we provide a parametric family of methods that allow the interpolation of the prefactor of the $\frac{N}{\lambda} $ term from $2$ to ($\, 1+\frac{1}{b}\,$) to obtain the best cost for different qubit availability regimes.
\end{abstract}
\title{Halving the cost of QROM}
\maketitle

\section{Introduction}
Quantum read only memory (QROM) provides the means for coherently accessing classical data on a quantum computer, and is the powerhouse of most quantum algorithms. It accounts for the majority of algorithmic resource requirements in Hamiltonian simulation \cite{berry2019qubitization,von2021quantum, lee2021even} via quantum signal processing \cite{low2017optimal,low2019hamiltonian, motlagh2024generalized, berry2024doubling}, state preparation \cite{babbush2018encoding, fomichev2024initial, rupprecht2026sparse}, unitary synthesis \cite{low2024trading, alonso2025quantum, kottmann2026parameter}, and quantum solvers for differential equations \cite{jennings2024cost, pocrnic2025constant, penuel2024detailed} and linear systems \cite{an2022quantum, dalzell2024shortcut, jennings2025randomized}. Therefore, any improvements in the efficiency of this fundamental algorithmic subroutine directly translates to reduction of overheads in almost all practical applications of quantum computers. Despite this, to the best of our knowledge, no improvements have been found for this key subroutine for over half a decade. \\

QROM aims to solve the problem of coherent table lookup. That is, given a ($\log_2 N$)-qubit address register in superposition, load a classical bitstring $f(x)$ of size $b$ corresponding to each of the $N$ possible computational basis states $\ket{x}$: $\sum_{x=0}^{N-1} \psi_x \ket{x}\ket{0} \to  \sum_{x=0}^{N-1} \psi_x \ket{x} \ket{f(x)}$. Naively, this can be implemented using $N$ back-to-back ($\log_2 N$)-controlled NOT gates with a total Toffoli cost of $N(\log_2 N - 1)$. This was further improved using the idea of unary iteration \cite{babbush2018encoding}, which ``streams" the bits of the one-hot encoding of the address register to reduce the Toffoli count to $N$. The next major improvement came with the advent of ``SelectSwap" techniques that allow for the utilization of ancillary qubits to reduce non-Clifford costs. In particular, the method introduced by Low et al. \cite{low2024trading} is able to reduce the cost of QROM to $\frac{N}{\lambda} + b\lambda$ using $b(\lambda-1)$ additional clean ancillary qubits. To avoid the need for the addition of a large number additional qubits, Ref. \cite{low2024trading} also introduces a construction based on dirty ancilla qubits with a Toffoli cost of $2\frac{N}{\lambda} + 4b\lambda$ using $b\lambda$ dirty ancilla qubits. Further improvements by Berry et al. \cite{berry2019qubitization} brought the Toffoli requirements to $2\frac{N}{\lambda} + 4b(\lambda - 1)$ and dirty ancilla counts to $b(\lambda-1)$, which still incurs a multiplicative factor of 2 due to the use of dirty ancillas. This has remained the state-of-the-art for the past seven years.\\

In this work we introduce a new construction for QROM that effectively halves its implementation cost. We first replace the ``SelectSwap" architecture with a ``SelectCopy" one to reduce the cost from $\,2\frac{N}{\lambda} + 4b(\lambda - 1)\,$ to $\,2\frac{N}{\lambda} + 2b(\lambda - 1) +2\lambda-6\,$ using the same $b(\lambda-1)$ dirty ancilla count. We then introduce an optimized algorithm for performing $m$ back-to-back QROMs with Toffoli cost $(m+1)\left(\frac{N}{\lambda} + b(\lambda - 1) +\lambda-3\right)$. We apply this new algorithm to the case of a single table lookup by treating a $b$-bit QROM as $\alpha$ back-to-back $\frac{b}{\alpha}$-bit QROMs, achieving a Toffoli count of $(1+\frac{1}{\alpha})\frac{N}{\lambda} + (b+\frac{b}{\alpha})(\alpha\lambda - 1) + (\alpha+1)(\alpha\lambda-3)$. This yields a parametric family of methods that interpolates the prefactor of the $\frac{N}{\lambda} $ term between $2$ to $(1+\frac{1}{b})$ at the cost of increasing subleading terms. This is particularly useful in regimes of practical interest, where the number of available dirty qubits is limited and $\frac{N}{\lambda} $ is the dominating term. When $N \gg b^2 \lambda^2$, setting $\alpha = b$ becomes optimal. This yields a total cost of $\sim (1+\frac{1}{b})\frac{N}{\lambda}$, reducing the cost by approximately $50\%$ and effectively matching the performance of clean-qubit QROM using dirty qubits.\\

\begin{figure}[t!]
    \centering
    \hspace*{-0.85\linewidth}
    \subfloat{
    \includegraphics[width=0.85\linewidth]{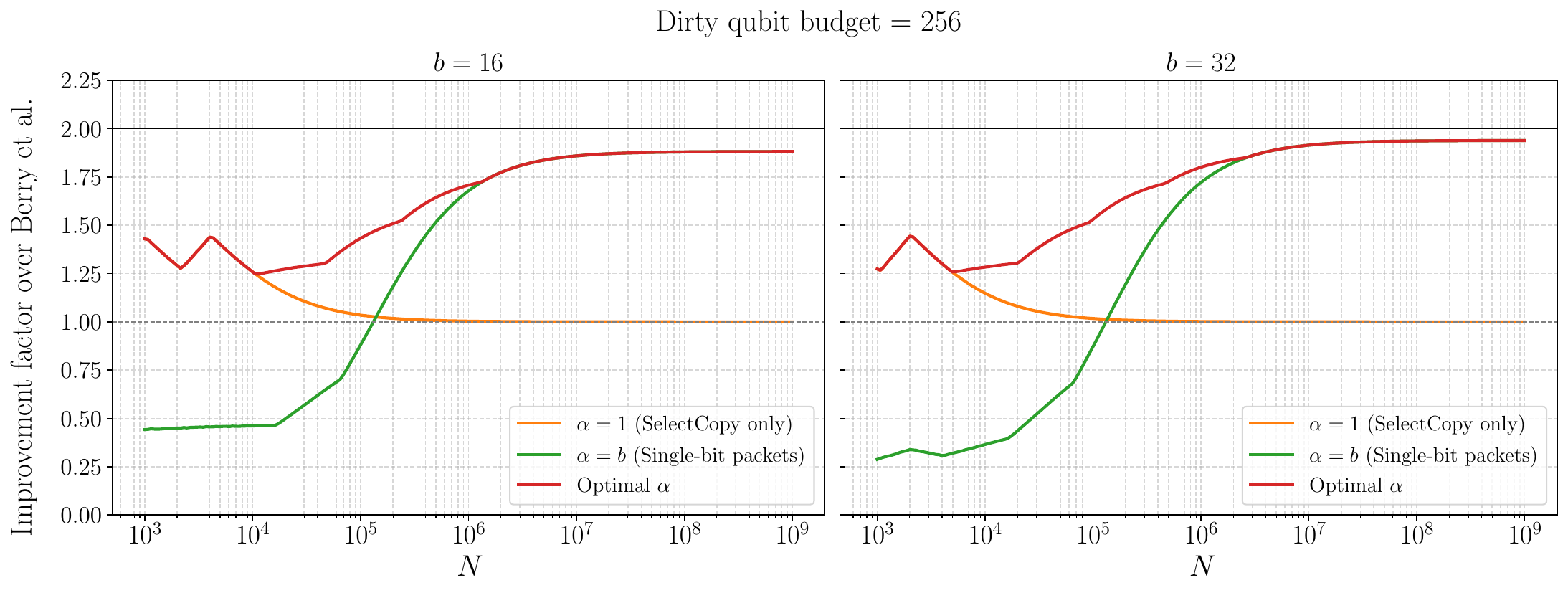}}
    
    \vspace{-0.8em}
    \hspace*{-0.85\linewidth}
    \subfloat{
    \includegraphics[width=0.85\linewidth]{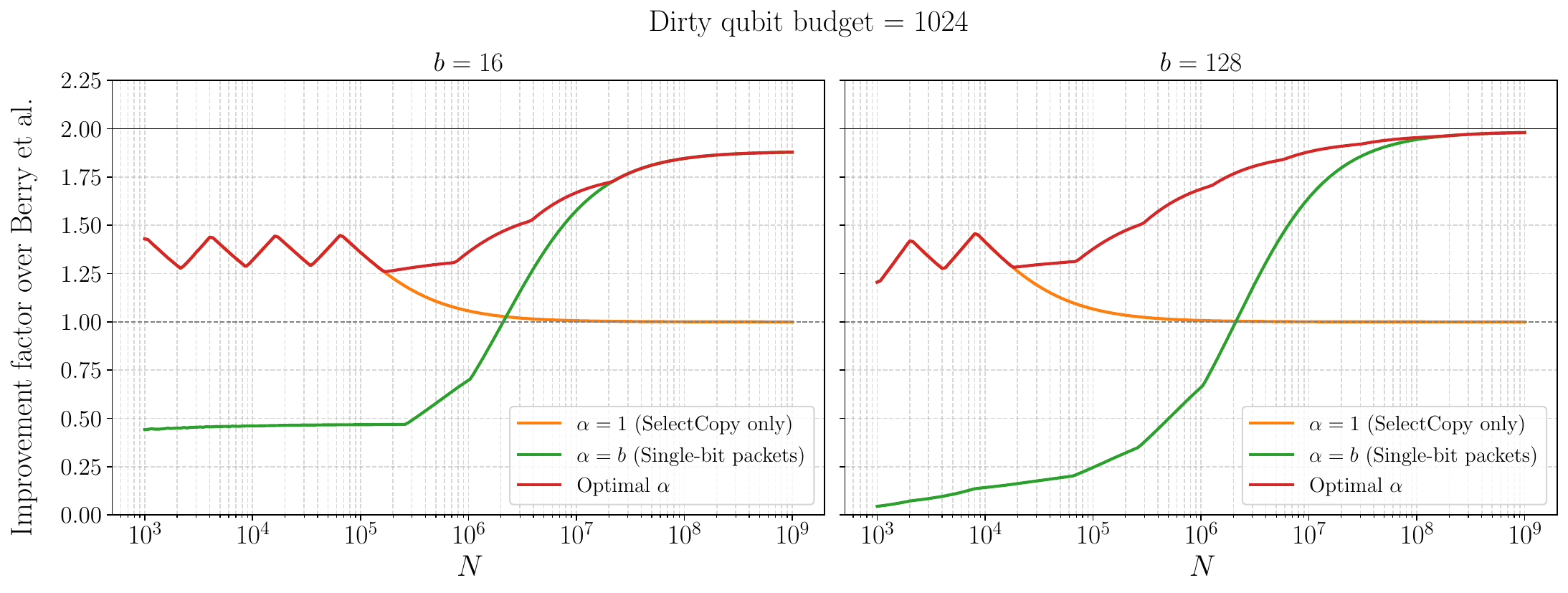}}
    \vspace{-0.8em}
    \caption{Cost comparison of our improved QROM against the previous state-of-the-art. We plot the improvement factor of our algorithm as a function of the number of elements $N$ being loaded for different bitstring length $b$ and number of available dirty qubits. We plot the improvement factor for the 2 ends of the spectrum $\alpha = 1$ (orange) and $\alpha = b$ (green) in the parametric family of methods provided in this work as well as for the optimal choice of $\alpha$ (red). The jagged lines at the beginning correspond to the qubit unconstrained regime where the optimal $\lambda$ does not yet make use of the entire dirty qubit budget.}
    \label{fig:qrom_comparison}
\end{figure}
\section{algorithm}
We now describe our improved QROM construction. For simplicity, in the rest of this section we assume $N$, $b$, and $\lambda$ are powers of 2; however, $N$ and $b$ can generally be arbitrary integers as shown in Appendix \ref{sec:non_power}.

\vspace{-0.2cm}
\subsection{Previous approaches}
QROM aims to solve the problem of coherent table lookup, that is, it implements
\begin{equation} \label{eq:data_load}
    \sum_{x=0}^{N-1} \psi_x \ket{x}\ket{0} \to  \sum_{x=0}^{N-1} \psi_x \ket{x} \ket{f(x)},
\end{equation}
where $f(x)$ is pre-computed classically. The high-level idea behind using ancilla qubits for speeding up QROM, as introduced in Ref. \cite{low2024trading}, is to realize that we can interpret the ($\log_2 N$)-qubit address register as two registers of size $\log_2 \lambda$ and $\log_2 N - \log_2 \lambda$
\begin{equation} \label{eq:data_load_2}
    \sum_{q=0}^{\frac{N}{\lambda}-1} \,\sum_{r=0}^{\lambda-1} \psi_{q,r} \ket{q}\ket{r}\ket{0} \to  \sum_{q=0}^{\frac{N}{\lambda}-1} \,\sum_{r=0}^{\lambda-1} \psi_{q,r} \ket{q}\ket{r}\ket{f(q,r)},
\end{equation}
where $x = q\cdot \lambda + r$. Then for each $\ket{q}$ we load all $\ket{f(q,r)}$ values for $r=0,...,\lambda -1$ in an ancillary register of size $b\lambda$ and swap the correct $\ket{f(q,r)}$ value into the first $b$ qubits based on the state of the $\ket{r}$ register.\\
\begin{figure}[t!]
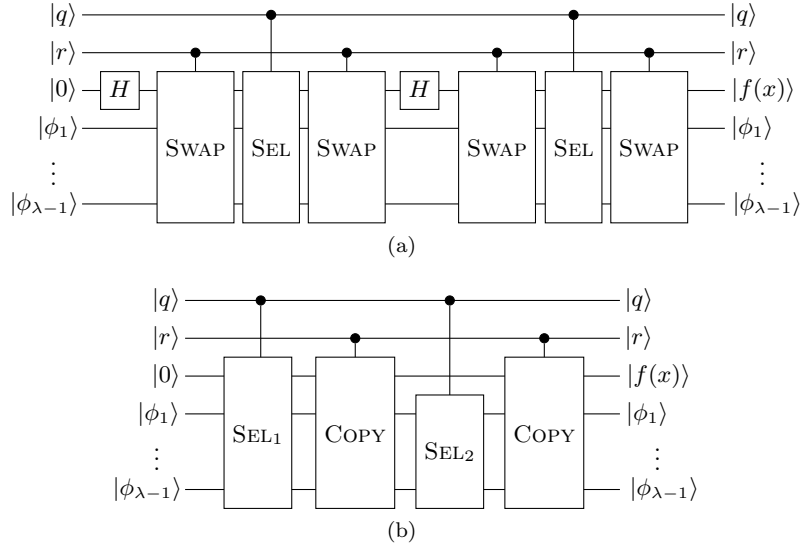

    \centering

    \subfloat[]{
    \begin{minipage}{1\textwidth}
    \centering
    \ctikzfig{berry_qroam}
    \end{minipage}}

    \subfloat[]{
    \begin{minipage}{1\textwidth}
    \centering
    \ctikzfig{qroam}
    \end{minipage}}

    \caption{Comparison of the (a) ``SelectSwap" QROM construction taken from Appendix A of Ref. \cite{berry2019qubitization} vs (b) our ``SelectCopy" construction. $\textsc{Sel}_1$ directly loads $f(q,0)$ into the clean register and $f(q,r) \oplus f(q,0)$ into the $r^{th}$ dirty register for $r=1,...,\lambda -1$, $\textsc{Sel}_2$ is similar but only goes over $r=1,...,\lambda -1$. This construction corresponds to $\alpha = 1$ in the parametric family of methods provided in Section \ref{sec:halving_qrom} of this work.}
    \label{fig:qrom_comparison_circ}
\end{figure}

Doing this with dirty ancillas uses the fact that for any quantum state $\ket{\phi}$, $\ket{\phi \oplus \, \phi \oplus \,f(q,r)} = \ket{f(q,r)}$, where $\oplus$ is defined as the bit-wise XOR product between 2 bitstrings. Hence, to replace the clean ancillas in the above approach with dirty ones, one would
\begin{enumerate}
    \item Swap the $r^{th}$ $b$-qubit ancilla register with the zeroth one based on the state of the $\ket{r}$ register.
    \item Copy the value of the zeroth register into a clean register.
    \item Swap the $r^{th}$ register with the zeroth one based on the state of the $\ket{r}$ register.
    \item Load all $f(q,r)$ values for $r=0,...,\lambda -1$ in the ancillary registers based on the state of the $\ket{q}$ register.
    \item Swap the $r^{th}$ $b$-qubit ancilla register with the zeroth one based on the state of the $\ket{r}$ register.
    \item Copy the value of the zeroth register into a clean register.
    \item Swap the $r^{th}$ register with the zeroth one based on the state of the $\ket{r}$ register.
    \item Unload all $f(q,r)$ values for $r=0,...,\lambda -1$ in the ancillary registers based on the state of the $\ket{q}$ register.
\end{enumerate}
This was the approach proposed in \cite{low2024trading} with Toffoli cost $2\frac{N}{\lambda} + 4b\lambda$ using $b\lambda$ dirty ancilla qubits. This was slightly refined in \cite{berry2019qubitization} to achieve Toffoli cost $2\frac{N}{\lambda} + 4b(\lambda - 1)$ and dirty ancilla counts of $b(\lambda-1)$. We now present our first improvement over this technique.

\begin{figure}[t!]
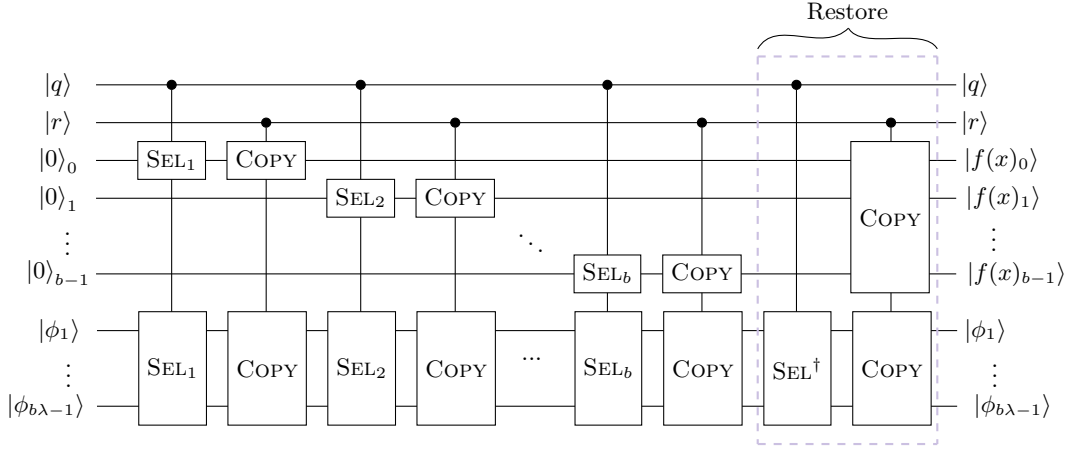

    \centering
    \ctikzfig{iterative_qrom} 
    \caption{Our QROM construction with sequential bit packets for the case of $\alpha = b$. Note that the dirty qubit registers, as well as the qubits in the output registers, are now depicted as single qubit registers. \textsc{Sel} and \textsc{Copy} operations are now loading and copying single qubit states into the output register per iteration. Recall that following $\textsc{Sel}_1$ operation, each subsequent $\textsc{Sel}_j$ loads the XOR between the new data and the prior bit in each dirty register. Finally, the ``Restore" operation refers to the final dirty-qubit clean-up, which is then copied into the output register to return the desired output. This operation is detailed in Appendix \ref{sec:restore}. This new construction is now available in PennyLane \cite{bergholm2018pennylane}.}
    \label{fig:one_bit_qrom}
\end{figure}

\subsection{Improvement with copy}\label{sec:select_copy}
Our first optimization comes from the realization that swapping the $r^{th}$ ancilla register with the zeroth register based on the state of the $\ket{r}$ register, copying the zeroth register into a clean register, and then swapping back is equivalent to copying the $r^{th}$ ancilla register into the clean register based on the state of the $\ket{r}$ register
\begin{align}
    \textsc{Copy}\ket{q}\ket{r}\ket{0}\bigotimes_j \ket{\phi_j}  &\to \ket{q}\ket{r}\ket{0\oplus\phi_r}\bigotimes_j \ket{\phi_j},\nonumber\\
    &= \ket{q}\ket{r}\ket{\phi_r}\bigotimes_j \ket{\phi_j},
\end{align}
and doing this multiplexed copy directly has half the Toffoli cost of 2 multiplexed swaps and requires no CNOTs (of which the previous approach requires $4b\lambda$). This already reduces the Toffoli cost of QROM from $\,2\frac{N}{\lambda} + 4b(\lambda - 1)\,$ to $\,2\frac{N}{\lambda} + 2b\lambda + 2\lambda - 6\,$, with the extra $2\lambda - 6$ Toffolis coming from having to do unary iteration over the $\ket{r}$ register for each multiplexed copy which has cost $\lambda - 3$.\\

However, it is possible to further reduce the Toffoli count of this approach to $\,2\frac{N}{\lambda} + 2b(\lambda-1) + 2\lambda - 6\,$ and dirty qubit count to $b(\lambda-1)$ using the following modification: instead of running over $r=0,...,\lambda -1$ when copying the $r^{th}$ register into the clean register, run over $r=1,...,\lambda -1$. Then modify the first load based on the state of the $\ket{q}$ register to directly load $f(q,0)$ into the clean register and $f(q,r) \oplus f(q,0)$ into the $r^{th}$ dirty register for $r=1,...,\lambda -1$, with a similar modification for the second load, but this time only load $f(q,r) \oplus f(q,0)$ into the $r^{th}$ dirty register for $r=1,...,\lambda -1$. We've included the circuit diagram for our ``SelectCopy"  based construction as well as the previous ``SelectSwap" construction for comparison in \cref{fig:qrom_comparison_circ}.\\

Note that our construction also reduces the uncomputation cost of table lookup using measurement based uncomputation from $\,2\frac{N}{\lambda'} + 4\lambda'\,$ Toffolis using $\lambda' - 1$ dirty qubits \cite{berry2019qubitization} to  $\,2\frac{N}{\lambda'} + 2\lambda'\, -6$ Toffolis by treating it as a QROM with $b=1$ where we load 1 for all indices of the address register that need a phase correction and replacing the copying of the $r^{th}$ dirty qubit based on the state of the $\ket{r}$ register with an application of Pauli Z on the $r^{th}$ dirty qubit based on the state of the $\ket{r}$ register.\\

\noindent We now describe our algorithm optimizing back-to-back applications of QROM which forms the basis for our approach in halving the cost of QROM.

\subsection{Sequential QROMs}
Sequential applications of QROM is a common occurrence in the design of quantum algorithms. An example of this is block-encoding algorithms for the electronic structure Hamiltonian in the tensor hypercontraction (THC) factorized form \cite{lee2021even}, where one multiplexes over different basis rotations, each consisting of a large number of Givens rotations, based on the state of a register encoding the coefficients of the terms in the Hamiltonian.\\

In such scenarios, the unloading step based on the state of the $\ket{q}$ register can be combined with the loading step for the next application of QROM by instead loading $f(q,r)_{j-1} \oplus f(q,r)_{j}$, where $f(q,r)_{j-1}$ is the value we would've needed to unload in the previous step. Using this trick, we only need to perform a single unloading step once at the end regardless of the number of sequential QROMs being applied. Therefore, the total number of loads (the coefficient multiplying $\frac{N}{\lambda}$ in our Toffoli count) is $(m+1)$ for $m$ sequential QROMs. With more thought, a similar reduction can be applied to the second multiplexed copy for the fix up $\ket{\phi_r \oplus f(q,r)} \to \ket{\phi_r \oplus f(q,r)\oplus \phi_r}=\ket{f(q,r)}$. The idea is to introduce another $b$-bit ancillary register that caches $\ket{\phi_r}$ once at the beginning using a multiplexed copy, and then uses that for all subsequent fix ups using only CNOT gates, giving us a total Toffoli cost of
\begin{equation}
    (m+1)\frac{N}{\lambda} + (m+2)\left(b\,(\lambda-1)+\lambda-3\right),
\end{equation}
for $m$ back-to-back QROMs. If the sequential QROMs are writing to a new clean register instead of re-writing the value from the previous QROM (which will be the case for the idea presented in the next section), this initial caching is not needed, as the fix up for all $m$ applications can be done at once at the end giving us a further reduction to
\begin{equation}\label{eq:sequential}
    (m+1)\left(\frac{N}{\lambda} + b(\lambda - 1) +\lambda-3\right).
\end{equation}

\subsection{Halving the cost of QROM with sequential bit packets} \label{sec:halving_qrom}
We are now ready to present our technique for halving the cost of QROM. The idea is simple: when constrained by the number of dirty qubits available to speed up QROM, the dominant term in cost of QROM becomes $2\frac{N}{\lambda}$; however, based on the construction provided in the previous section for performing $m$ sequential QROMs, we find that the prefactor for this term becomes $(m+1)$ as opposed to $2m$. Hence we will utilize this fact to reduce the prefactor in front of $\frac{N}{\lambda}$ by treating a $b$-bit QROM as $\alpha$ sequential $\frac{b}{\alpha}$-bit QROMs. Without loss of generality, assume $N$, $b$, and $\alpha$ are powers of 2. This means that if we previously had $\sim b\lambda$ dirty qubits available, allowing for a ``SelectCopy" depth of $\lambda$, with $\frac{b}{\alpha}$-bit QROM we can now have a ``SelectCopy" depth of $\alpha\lambda$. Subbing this into the equation of \cref{eq:sequential} along with our new bit size of $\frac{b}{\alpha}$, we get that the Toffoli cost of $\alpha$ sequential $\frac{b}{\alpha}$-bit QROMs is
\begin{equation}\label{eq:bit_batch_load}
    (1+\frac{1}{\alpha})\frac{N}{\lambda} + (b+\frac{b}{\alpha})(\alpha\lambda - 1) + (\alpha+1)(\alpha\lambda-3),
\end{equation}
using $b\lambda - \frac{b}{\alpha}$ dirty ancillas. Setting $\alpha = 1$ recovers the cost from \cref{sec:select_copy}, while $\alpha = b$ achieves Toffoli cost
\begin{equation}
    (1+\frac{1}{b})\frac{N}{\lambda} + 2(b+1)(b\lambda - 2),
\end{equation}
which when $N \gg b^2 \lambda^2$ gives the cost
\begin{equation}
    \sim(1+\frac{1}{b})\frac{N}{\lambda}.
\end{equation}
This leads to a reduction of approximately $50\%$  over the previous state-of-the-art, and effectively matches the performance of clean-qubit QROM using dirty qubits. However, more generally $\alpha$ is a tunable parameter that can be set based on the value of $N$ and number of available dirty qubits to find the optimal Toffoli count for particular cases of interest. We provide a cost comparison between our improved method and the previous state-of-the-art in \cref{fig:qrom_comparison}.\\

\section{Conclusion}
Table lookup is one of the most widely used algorithmic subroutines in fault-tolerant quantum algorithms and accounts for the majority share of algorithmic overheads in almost all practical applications of quantum computers. Therefore, any reductions in its implementation cost has far reaching implications for the fields of quantum algorithms and applications. In this work, we provided the first major improvement to this foundational subroutine in over half a decade. By introducing a number of advanced yet intuitive techniques, we are able to effectively halve the cost QROM for regimes of practical interest and match the performance clean-qubit QROM using dirty qubits.\\

A future direction worth considering is the application of integer division as a pre-processing step to lift the restriction on $\lambda$ having to be a power of 2 \cite{low2024trading} in order to make full use of available dirty qubits when the budget cannot be written as $2^k  b$ for any integer $k$. This becomes particularly important in many practical regimes where qubit constraints prevent one from achieving the optimal $O(\sqrt{N})$ Toffoli scaling \cite{low2024trading}, and maximal use of available dirty qubits is in the best interest of algorithm designers. With the recent improvements in cost of integer division \cite{mukhopadhyay2026improved}, the incorporation of this idea with the framework presented here could potentially provide further reductions in specific regimes of interest.

\newpage

\bibliography{bib}
\newpage

\appendix

\section{Non-power of 2 bit-length}\label{sec:non_power}
Let $\mu$ be the number of bits loaded in each of the $\alpha$ groups/iterations. In more general applications of this algorithm, the problem parameters will not be powers of 2. In this section we lift the power of 2 constraint on $b, \mu$, and $N$. In this case, it can be more natural to describe the costs of this algorithm in terms of the number of bits we load per iteration $\mu$ rather than the number of iterations we perform $\alpha$. If we load $\mu$ bits per iteration, then clearly we need to iterate $\alpha = \ceil{b/\mu}$ times. For $\floor{b/\mu}$ iterations $\mu$ bits are loaded, and in the final iteration $(b \, \mathrm{mod} \, \mu)$ bits are loaded, which reduces the cost of the final copy operation. With this intuition, we provide the Toffoli cost in the following Theorem. 

\begin{theorem}\label{thm:mu_bit_qrom}
    Given a boolean function $f :  \mathbb{Z}_N \to \mathbb{Z}_2^b$ and $\lambda$ a power of 2 satisfying $1 < \lambda < N$, then a circuit to perform the transformation in Equation \ref{eq:data_load} can be constructed using 
    \begin{equation} \label{eq:non_2_cost}
        (\ceil{b/\mu}+1)\Bigr(\ceil{N/\lambda}+\lambda -3 \Bigr) + (\lambda-1)\Bigr( \mu (\floor{b/\mu}+1) +b \, \mathrm{mod}\, \mu \Bigr)
    \end{equation}
    Toffoli gates, $\mu(\lambda-1)$ dirty ancilla, $\max\{\ceil{\log(N/\lambda )}, \log(\lambda)\}$ clean ancilla qubits, and $b$ clean qubits to store the output.
\end{theorem}
\begin{proof}
    For compactness let $f_{x,j}$ be the $j$th bit of the $b$-bit string $f_x:=f(x)$ that we wish to load. Also, let $c_{q\cdot \lambda+r,j}:= f_{q\cdot \lambda+r,j}\oplus f_{q\cdot \lambda,j}$ which is just the XOR between each data bit with the $(q\cdot  \lambda)$-th bit that is loaded in the top most register. Suppressing the $\ket q$ and $\ket r$ registers which do not change, and considering the non-trivial case where $r\neq 0$, the circuit then acts via the following steps:
    \begin{align}
    &\ket 0^b \bigotimes_{l=1}^{\lambda-1} \ket{\phi_l} \xrightarrow[]{\textsc{Sel}_1} \ket{f_{q\cdot \lambda,0}}...\ket{f_{q\cdot \lambda,\mu-1}} \ket 0^{b-\mu} \bigotimes_{l=1}^{\lambda-1} \bigotimes_{j=0}^{\mu-1} \ket{\phi_l \oplus c_{q\cdot \lambda +l,j}} \\
    &\xrightarrow[]{\textsc{Copy}} \ket{f_{q\cdot \lambda,0} \oplus c_{q\cdot \lambda +r,0} \oplus \phi_{r,0}}...\ket{f_{q\cdot \lambda,\mu-1} \oplus c_{q\cdot \lambda +r,\mu-1} \oplus \phi_{r,\mu-1}} \ket 0^{b-\mu} \bigotimes_{l=1}^{\lambda-1} \bigotimes_{j=0}^{\mu-1} \ket{\phi_{l,j} \oplus c_{q\cdot \lambda +l,j}} \\
    &= \ket{ f_{q\cdot \lambda +r,0} \oplus \phi_{r,0}}...\ket{f_{q\cdot \lambda +r,\mu-1}\oplus \phi_{r,\mu-1}} \ket 0^{b-\mu} \bigotimes_{l=1}^{\lambda-1} \bigotimes_{j=0}^{\mu-1} \ket{\phi_{l,j} \oplus c_{q\cdot \lambda +l,j}}  \\
    & \xrightarrow[]{ \textsc{Copy} \;\textsc{Sel}_2} \ket{f_{q\cdot \lambda +r,0} \oplus \phi_{r,0}}...\ket{f_{q\cdot \lambda +r,2\mu-1}\oplus \phi_{r,\mu-1}}  \ket 0^{b-2\mu} \bigotimes_{l=1}^{\lambda-1} \bigotimes_{j=0}^{\mu-1} \ket{\phi_{l,j} \oplus c_{q\cdot \lambda +l,j+\mu}} \\
    & \vdots \notag \\
    &\xrightarrow[]{ \textsc{Copy} \;\textsc{Sel}_{\floor{b/\mu}}} \ket{f_{q\cdot \lambda +r,0} \oplus \phi_{r,0}}...\ket{f_{q\cdot \lambda +r,\mu \floor{b/\mu} -1}\oplus \phi_{r,\mu-1}} \ket{0}^{(b \, \mathrm{mod} \, \mu)} \bigotimes_{l=1}^{\lambda-1} \bigotimes_{j=0}^{\mu-1} \ket{\phi_{l,j} \oplus c_{q\cdot \lambda +l,j+\mu\floor{b/\mu}-1}}.
\end{align}
At this point we have loaded the first $b- (b \, \mathrm{mod} \, \mu) = \mu \floor{b/\mu}$ bits of the data. If $b/\mu \in \mathbb{Z}$ then we are done. If not, we have one more operation that loads $(b \, \mathrm{mod} \, \mu)$ bits into $\lambda$ registers of that size (ignoring the extra $\mu - (b \, \mathrm{mod} \, \mu)$ available qubits in each of the $\lambda-1$ dirty qubit registers), giving $\ceil{b/\mu}$ iterations in total. For the restore operation, $\textsc{Sel}^\dagger$ simply loads the corrections to return the dirty qubit register to its original state $\bigotimes_{l=1}^{\lambda-1}\ket{\phi_l}$. In reference to Figure \ref{fig:restore}, the final \textsc{Copy} initializes one temp-AND for each dirty qubit state, controlled on the unary iteration over $\ket r$ and on said dirty qubit state $\ket{\phi_{r,[0, \mu-1]}}$. Controlled on this temporary qubit, we apply a CNOT to every location where this error occurs in the output register. This has the same Toffoli cost of the prior \textsc{Copy} operators at $(\lambda-1)\mu$. Compiling the cost of all select and copy operations yields the final cost. The cost pertaining to unary iteration over $\ket q$ and $\ket r$ are taken to be $(\ceil{N/\lambda}-1)$ and $(\lambda-2)$ respectively based on Ref. \cite{babbush2018encoding}.  
\end{proof}

In Equation \ref{eq:non_2_cost}, the first term is solely the cost of unary iteration over both registers, and the second term comes from the Toffoli gates required to copy qubit states. The costs used are for controlled unary iteration, but in principle \textsc{Copy} operations need not be controlled since the Restore operation serves as the adjoint. Furthermore, we derived Equation \ref{eq:non_2_cost} using $\mu(\lambda-1)$ dirty qubits, which can be significantly less than the aforementioned QROM with $b(\lambda -1)$. To adjust the construction based on loading batches of $\mu$ bits such that we use a comparable number of qubits, we can make the rescaling $\lambda' \to \lambda \floor{b/\mu}$. This upper bounds the dirty qubit usage by $(\lambda \floor{b/\mu}-1)\mu \leq \lambda b - \mu$, which compared to the prior construction uses at most $(b-\mu)$ additional qubits. Plugging in this rescaled $\lambda$ then gives
\begin{equation}
    (\ceil{b/\mu}+1)\biggr(\ceil{\frac{N}{\lambda\floor{b/\mu}}}+\lambda\floor{b/\mu} -3 \biggr) + (\lambda\floor{b/\mu}-1)\Bigr( \mu (\floor{b/\mu}+1) + b \, \mathrm{mod}\, \mu \Bigr), 
\end{equation}
which yields a leading factor $\approx N/\lambda$ as previously claimed. Recall that scaling $\lambda$ in this way may result in it not being a power of 2 (requiring rounding to the nearest power of 2 or applying integer division). Returning to the scenario considered in the main text, if we enforce that all parameters be a power of 2, and note that $\alpha = b/\mu$, then the expression above simplifies to exactly match the claim of Equation \ref{eq:bit_batch_load}. 

\section{Restore operation} \label{sec:restore}
An important part of the construction is realizing that the dirty qubit states that XOR the output bits repeat based on the choice of $\alpha$ or $\mu$. For example, if we load an 5-bit string using 3 iterations of 2 bits, prior to \textsc{Restore} we have:
\begin{equation}
    \ket{f_{x,0} \oplus \phi_{r, 0} }\ket{f_{x,1} \oplus \phi_{r, 1} }\ket{f_{x,2} \oplus \phi_{r, 0} }\ket{f_{x,3} \oplus \phi_{r, 1} }\ket{f_{x,4} \oplus \phi_{r, 0} }.
\end{equation}
The restore works by iterating over $\ket r$ and storing the state of $\phi_{r, j}$ in a temporary ancilla, and then applying \textsc{Cnot} to each qubit where this error occurs, which is deterministic based on the circuit construction (see Figure \ref{fig:restore}).

\begin{figure}[htbp!]
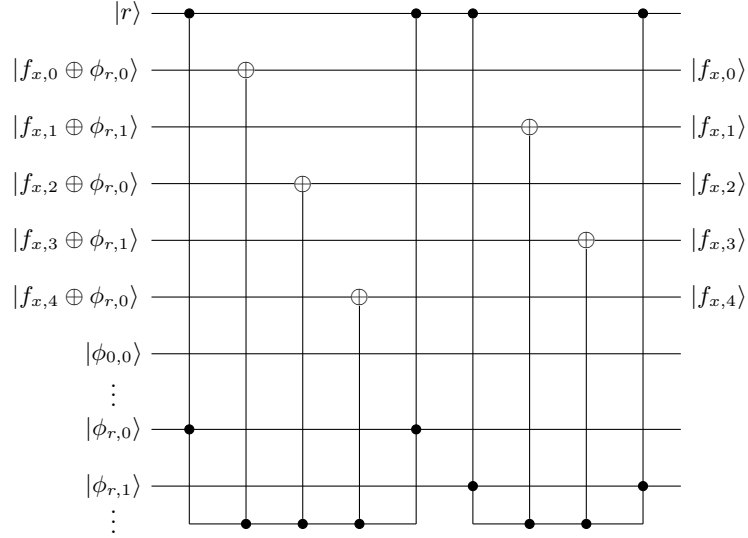

    \centering
    \ctikzfig{restore} 
    \caption{A circuit construction of the \textsc{Copy} operation required to perform the \textsc{Restore} operation highlighted outlined in Figure \ref{fig:one_bit_qrom}. The $\textsc{Sel}^\dagger$ operation simply unloads the previous data loaded into the dirty qubits, and the final copy uses the dirty qubit registers to clean up the data register as shown. The control on the $\ket r$ register is taken to be a single qubit control on the corresponding sawtooth wire from the unary iteration circuit in Figure 10 of Ref. Of course, in the full operation, $\mu$ temp-ANDs must be initialized for each outcome from unary iteration \cite{babbush2018encoding}.}
    \label{fig:restore}
\end{figure}

\end{document}